\documentclass[a4,a4wide]{article}
\usepackage{aaai}
\usepackage{times}
\usepackage{helvet}
\usepackage{courier}
\usepackage{paralist}
\usepackage{tikz}
\usepackage{url}
\usepackage{verbatim}

\newcommand{\eqref}[1]{(\ref{#1})}

\newcommand{\at}{{At}}
\newcommand{\Lit}{{Lit}}

\newcommand{\dnot}{\mathit{not} \;}

\newcommand{\lrule}[2]{#1 \leftarrow #2}
\newcommand{\head}[1]{head(#1)}
\newcommand{\body}[1]{body(#1)}
\newcommand{\bodyp}[1]{body^+(#1)}
\newcommand{\bodym}[1]{body^-(#1)}

\newcommand{\reduct}[2]{{{#1}^{#2}}}

\newcommand{\symminpos}{\mathcal{Q}}
\newcommand{\minpos}[1]{{\symminpos(#1)}}

\newcommand{\as}[1]{{\mathcal{AS}({#1})}}
\newcommand{\agras}[2]{{G_{#2}(#1)}}

\newcommand{\lpp}[1]{\mathcal{#1}}

\newcommand{\alrule}[2]{#1 & \leftarrow & #2}

\newcommand{\laarule}[3]{\larule{#1}{#2}{#3}{0.5em}{0.5em}{0.5em}}
\newcommand{\larule}[6]{\ifthenelse{\isempty{#1}}{}{$#1$: \hspace{#4} \=} $#2$ \hspace{#5} \= $\leftarrow$ \hspace{#6} \= $#3$}
\newcommand{\lbrule}[3]{\ifthenelse{\isempty{#1}}{}{$#1$: \>} $#2$ \> $\leftarrow$ \> $#3$}
\newcommand{\tset}[1]{\hspace{#1}\=}
\newcommand{\tgo}{\>}

\newcommand{\rname}[1]{n_{#1}}
\newcommand{\rlit}[2]{{#1}^{#2}}
\newcommand{\trules}[2]{{T_{#1}^{#2}}}
\newcommand{\namess}[2]{{N_{#2}({#1})}}
\newcommand{\transr}[2]{{t_{#2}({#1})}}
\newcommand{\trans}[1]{{t({#1})}}

\newcommand{\frag}[1]{{F({#1})}}
\newcommand{\exrules}[1]{{R({#1})}}

\newcommand{\sn}[2]{\mathcal{PAS}_{#1}^{#2}}

\newcommand{\snpasd}{\sn{D}{}}
\newcommand{\snpasindno}{\sn{GNO}{}}
\newcommand{\snpasindgen}{\sn{G}{}}

\newcommand{\snpasdst}{\sn{DST}{}}
\newcommand{\snpaswzl}{\sn{WZL}{}}
\newcommand{\snpasbe}{\sn{BE}{}}

\newcommand{\pasd}[1]{\snpasd(#1)}
\newcommand{\pasindno}[1]{\snpasindno(#1)}
\newcommand{\pasindgen}[1]{\snpasindgen(#1)}

\newtheorem{definition}{Definition}
\newtheorem{proposition}{Proposition}

\newtheorem{example}{Example}
\newtheorem{notation}{Notation}

\newtheorem{corollary}{Corollary}

\newtheorem{principle}{Principle}{\itshape}{\rmfamily}

\usepackage{xifthen}
\newcounter{counter:tmp}

\newenvironment{proof}{\noindent\emph{Proof:}}{}

\title{A Family of Descriptive Approaches To Preferred Answer Sets}
\author{Alexander \v{S}imko\\
Department of Applied Informatics\\
Faculty of Mathematics, Physics and Informatics\\
Comenius University in Bratislava\\
Mlynsk\'{a} dolina, 842 48 Bratislava, Slovakia}

\begin{document}
\nocopyright
\maketitle
\begin{abstract}
In logic programming under the answer set semantics, preferences on rules are used to choose which of the conflicting rules are applied. Many interesting semantics have been proposed. Brewka and Eiter's Principle \ref{principle:i} expresses the basic intuition behind the preferences. All the approaches that satisfy Principle \ref{principle:i} introduce a rather imperative feature into otherwise declarative language. They understand preferences as the order, in which the rules of a program have to be applied. In this paper we present two purely declarative approaches for preference handling that satisfy Principle I, and work for general conflicts, including direct and indirect conflicts between rules. The first approach is based on the idea that a rule cannot be defeated by a less preferred conflicting rule. This approach is able to ignore preferences between non-conflicting rules, and, for instance, is equivalent with the answer set semantics for the subclass of stratified programs. It is suitable for the scenarios, when developers do not have full control over preferences. The second approach relaxes the requirement for ignoring conflicting rules, which ensures that it stays in the NP complexity class. It is based on the idea that a rule cannot be defeated by a rule that is less preferred or depends on a less preferred rule. The second approach can be also characterized by a transformation to logic programs without preferences. It turns out that the approaches form a hierarchy, a branch in the hierarchy of the approaches by Delgrande et. al., Wang et. al., and Brewka and Eiter. Finally, we show an application for which the existing approaches are not usable, and the approaches of this paper produce expected results.
\end{abstract}

\section{Introduction}

Preferences on rules are an important knowledge representation concept. In logic programming, one usually writes general rules, and needs to express exceptions. Consider we have the following rules
\begin{tabbing}
\laarule{r_1}{select(car_1)}{nice(car_1)}\\
\lbrule{r_2}{\neg select(car_1)}{expensive(car_1)}\\
\lbrule{r_3}{select(car_1)}{fast(car_1)}
\end{tabbing}
If a $car_1$ is both nice, expensive, and fast, the rules lead to contradiction. If we have preferences on rules, e.g., we prefer $r_1$ over $r_2$, and $r_2$ over $r_3$, we can use default negation to express exceptions between rules. Since the rules $r_1$ and $r_3$ have the same head, we have to use an auxiliary literal in order to ensure that $r_3$ does not defeat $r_2$. \begin{tabbing}
\larule{r_{1a}}{aux}{nice(car_1)}{0.25em}{4.5em}{0.5em}\\
\lbrule{r_{1b}}{select(car_1)}{select(car_1)}\\
\lbrule{r_2}{\neg select(car_1)}{expensive(car_1), \dnot aux}\\
\lbrule{r_3}{select(car_1)}{fast(car_1), \dnot \neg select(car_1)}
\end{tabbing}
The hand-encoding of preferences has to use auxiliary literals, we have to split rules, and the resulting program is less readable. If the complementary literals are derived via other rules, and the program has hundreds of rules, the hand-encoding becomes even less readable.

More readable way to encode the exceptions between the rules is to make rules mutually exclusive, represent preferences using a relation on rules, and use a semantics for logic programs with preferences, in order to handle preferences.
\begin{tabbing}
\larule{r_1}{select(car_1)}{nice(car_1), \dnot \neg select(car_1)}{0em}{0.5em}{0em}\\
\lbrule{r_2}{\neg select(car_1)}{expensive(car_1), \dnot select(car_1)}\\
\lbrule{r_3}{select(car_1)}{fast(car_1), \dnot \neg select(car_1)}\\
\\
$r_3 < r_2 < r_1$
\end{tabbing}
The rules $r_1$ and $r_2$ are mutually exclusive: whenever we apply the rule $r_1$, the rule $r_2$ is not applicable, and vice versa. We call this mutual exclusivity a conflict. The resulting program is much tolerant to changes. If we decide that the rule $r_3$ is the most preferred, and $r_3$ is the least preferred, only the preference relation needs to be changed, and the rules stay intact.

Several semantics for logic programs with preferences on rules have been proposed in the literature. In the first group  are semantics that extend the well-founded semantics \cite{VanGelder:1991wl}: \cite{Brewka:1996vs,Wang:2000tq,Schaub:2002tr} modify the alternating fixpoint characterization of the well-founded semantics in order to take preferences into account.

In the second group are the semantics that extend the answer set semantics \cite{Gelfond:1991wx}. Each model of a program with preferences, called a preferred answer set, is guaranteed to be an answer set of the underlying program without preferences. \cite{Brewka:1999uv,Wang:2000tq,Delgrande:2002uc} provide prescriptive \cite{Delgrande:2004um} semantics, i.e. preferences are understood as the order in which the rules of a program have to be applied. A rule can be defeated only by rules that were applied before it w.r.t. to this order. Each answer set is tested whether it can be constructed in aforementioned way. \cite{Zhang:1997tc} iteratively non deterministically removes from a program less preferred rules that are defeated by the remainder of the program. \cite{Sakama:2000wo} transforms preferences on rules to preferences on literals, which leads to comparison of the sets of generating rules. Roughly speaking, answer set generated by maximal rules (w.r.t. a preference relation) are selected. \cite{Sefranek:2008wo} understands preference handling as a kind of argumentation.

Brewka and Eiter have proposed Principle \ref{principle:i} \cite{Brewka:1999uv} that captures the intuition behind preferences on rules. If two answer sets are generated by the same rules except for two rules, and one rule is preferred over the other, an answer set generated by the less preferred rule should not be preferred.

The existing approaches to preference handling that satisfy Principle \ref{principle:i} \cite{Brewka:1999uv,Wang:2000tq,Delgrande:2002uc}, denoted here as $\snpasbe$, $\snpaswzl$ and $\snpasdst$, introduce a rather imperative feature into the otherwise declarative language. They understand preferences on rules as the order in which the rules of a program have to be applied. This, on the one hand goes against declarative spirit of logic programming. On the other hand, it makes the approaches unusable in the situations when we need to automatically generate preferences.

\begin{example}
\label{example:motivation}
Consider a modified version of the scenario from \cite{Brewka:1999uv}. Imagine we have a car recommender system. A program written by the developers of the system contains a database of cars and recommends them to a user.
\begin{tabbing}
\laarule{r_1}{nice(car_1)}{}\\
\lbrule{r_2}{safe(car_2)}{}\\
\\
\lbrule{r_3}{rec(car_1)}{nice(car_1), \dnot \neg rec(car_1)}\\
\lbrule{r_4}{rec(car_2)}{nice(car_2), \dnot \neg rec(car_2)}
\end{tabbing}
The system recommends nice cars to the user. We allow the user to write his/her own rules during the run time of a system. Imagine the user writes the following rules
\begin{tabbing}
\laarule{u_1}{\neg rec(car_2)}{rec(car_1)}\\
\lbrule{u_2}{\neg rec(car_1)}{rec(car_2)}\\
\\
\lbrule{u_3}{rec(car_1)}{safe(car_1), \dnot \neg rec(car_1)}\\
\lbrule{u_4}{rec(car_2)}{safe(car_2), \dnot \neg rec(car_2)}
\end{tabbing}
to say that maximally one car should be recommended, and that the user is interested in safe cars.

Due to the rules $u_1$ and $u_2$, the rule $u_3$ is conflicting with $r_4$:
\begin{inparaenum}[(i)]
\item The rule $u_1$ depends on $r_3$, and its head is in the negative body of $u_4$.
\item The rule $u_2$ depends on $u_4$, and its head is in the negative body of $r_3$.
\end{inparaenum}
We also have that $u_3$ is conflicting with $u_4$, and $r_3$ is conflicting with $r_4$ and $u_4$. All the conflicts are indirect -- without the rules $u_1$ and $u_2$ there are no conflicts.

The purpose of the user's rules is to override the default behaviour of the system in order to provide the user the best experience possible. Therefore we want the rule $u_3$ to override $r_4$, and $u_4$ to override $r_3$. Since the $u_i$ rules are only known at the run time, preferences cannot be specified beforehand by the developers of the system. Moreover, we cannot expect a user to know all the $r_i$ rules. It is reasonable to prefer each $u_i$ rule over each $r_j$ rule, and let the semantics to ignore preferences between non-conflicting rules. Hence we have the preferences:\\
\\
$u_1$ is preferred over $r_1$\\
$u_1$ is preferred over $r_2$\\
$\dots$\\
$u_4$ is preferred over $r_4$\\
\\
The prerequisites $nice(car_2)$ and $safe(car_1)$ of $r_4$ and $u_3$ cannot be derived. The only usable conflicting rules are $r_3$ and $u_4$. The rule $u_4$ being preferred, $u_4$ defines an exception to $r_3$. We expect $u_4$ to be applied, and $r_3$ defeated. The only answer set that uses $u_4$ is $S = F \cup \{ \neg rec(car_1), rec(car_2) \}$ where $F = \{ nice(car_1), safe(car_2)\}$. Hence $S$ is the unique expected preferred answer set.

None of the existing approaches satisfying Principle \ref{principle:i} works as expected. $\snpasbe$ does not handle indirect conflicts, and provides two preferred answer sets $S$ and $S_2 = F \cup \{ rec(car_1), \neg rec(car_2) \}$. $\snpasdst$ and $\snpaswzl$ provide no preferred answer set due to they imperative nature. Since $u_4$ is preferred over $r_2$, they require that $u_4$ is applied before $r_2$. It is impossible as $r_2$ is the only rule that derives $r_4$'s prerequisite.

It is not crucial for the example that the facts $r_1$ and $r_2$ are less preferred. If one feels that they should be separated from the rest of the rules, we can easily modify the program, e.g., by replacing the fact $safe(car_2)$ by the fact $volvo(car_2)$ and the rule $\lrule{safe(car_2)}{volvo(car_2)}$.
\end{example}


Our goal is to develop an approach to preference handling that
\begin{inparaenum}[(i)]
\item is purely declarative,
\item satisfies Brewka and Eiter's Principle \ref{principle:i}, and
\item is usable in the above-mentioned situation.
\end{inparaenum}

We have already proposed such a semantics for the case of direct conflicts, and we denote it by $\snpasd$ \cite{Simko:2013tm}. We understand this semantics as the reference semantics for the case of direct conflicts, and extend it to the case of general conflicts in this paper.

We present two approaches. The first one, denoted by $\snpasindgen$, is based on the intuition that \emph{a rule cannot be defeated by a less preferred (generally) conflicting rule}. The approach is suitable for situations when we need to ignore preferences between non-conflicting rules, and is equivalent to the answer set semantics for the subclass of stratified programs. We consider this property to be important for the aforementioned situations as stratified programs contain no conflicts.

The second approach, denoted $\snpasindno$, relaxes the requirement for ignoring preferences between non-conflicting rules, and stays is the NP complexity class. There are stratified programs with answer sets and no preferred answer sets according to the approach. The approach is suitable in situations when a developer has a full control over a program. The approach is based on the intuition that \emph{a rule cannot be defeated by a less preferred rule or a rule that depends on a less preferred rule}. The approach can be also characterized by a transformation from logic programs with preferences to logic programs without preferences such that the answer sets of the transformed program (modulo new special-purpose literals) are the preferred answer sets of an original one.

The two approaches of this paper and our approach for direct conflicts $\snpasd$ form a hierarchy, which in general does not collapse. Preferred answer sets of $\snpasindno$ are preferred according to $\snpasindgen$, and preferred answer sets of $\snpasindgen$ are preferred according to $\snpasd$.

$\snpasd$ is thus the reference semantics for the case of direct conflicts. $\snpasindno$ can be viewed as a computationally acceptable approximation of $\snpasindgen$. $\snpasindno$ is sound w.r.t. $\snpasindgen$, but it is not complete w.r.t. $\snpasindgen$, meaning that each preferred answer set according to $\snpasindno$ is a preferred answer set according to $\snpasindgen$, but not vice versa.

When dealing with preferences, it is always important to remember what the abstract term ``preferences'' stands for. Different interpretations of the term lead to different requirements on a semantics. We want to stress that we understand preferences as a mechanism for encoding exceptions between rules in this paper.

The rest of the paper is organized as follows. We first recapitulate preliminaries of logic programming, answer set semantics and our approach to preferred answer sets for direct conflicts $\snpasd$. Then we provide the two approaches to preferred answer sets for general conflicts. After that we show relation between the approaches of this paper, and also between approaches of this paper and existing approaches. Finally we show how the approaches work on the problematic program from Example \ref{example:motivation}. Proofs not presented here can be found in the technical report \cite{simko:report}.

\section{Preliminaries}
\label{section:preliminaries}

In this section, we give preliminaries of logic programming and the answer set semantics. We recapitulate the alternative definition of answer sets based on generating sets from \cite{Simko:2013tm}, upon which this paper builds.

\subsection{Syntax}

Let $\at$ be a set of all atoms. A \emph{literal} is an atom or an expression $\neg a$, where $a$ is an atom. Literals of the form $a$ and $\neg a$ where $a$ is an atom are \emph{complementary}. A \emph{rule} is an expression of the form
$\lrule{l_0}{l_1, \dots, l_m, \dnot l_{m+1}, \dots, \dnot l_n}$, where $0 \leq m \leq n$, and each $l_i$ ($0 \leq i \leq n)$ is a literal. Given a rule $r$ of the above form we use $\head{r} = l_0$ to denote the \emph{head} of $r$, $\body{r} = \{ l_1, \dots, \dnot l_n \}$ the \emph{body} of $r$. Moreover, $\bodyp{r} = \{l_1, \dots, l_m\}$ denotes the \emph{positive} body of $r$, and $\bodym{r} = \{l_{m+1}, \dots, l_n\}$ the \emph{negative} body of $r$. For a set of rules $R$, $\head{R} = \{ \head{r} : r \in R \}$. A \emph{fact} is a rule with the empty body. A \emph{logic program} is a finite set of rules.

We say that a rule $r_1$ defeats a rule $r_2$ iff $\head{r_1} \in \bodym{r_2}$. A set of rules $R$ defeats a rule $r$ iff $\head{R} \cap \bodym{r} \neq \emptyset$. A set of rules $R_1$ defeats a set of rules $R_2$ iff $R$ defeats a rule $r_2 \in R_2$.

For a set of literals $S$ and a program $P$ we use $\agras{S}{P} = \{ r \in P : \bodyp{r} \subseteq S \mbox{ and } \bodym{r} \cap S = \emptyset \}$.

A \emph{logic program with preferences} is a pair $(P,<)$ where:
\begin{inparaenum}[(i)]
\item $P$ is a logic program, and
\item $<$ is a transitive and asymmetric relation on $P$.
\end{inparaenum}
If $r_1 < r_2$ for $r_1,r_2 \in P$ we say that $r_2$ is \emph{preferred over} $r_1$.

\subsection{Answer Set Semantics}

A set of literals $S$ is consistent iff $a \in S$ and $\neg a \in S$ holds for no atom $a$.

A set of rules $R \subseteq P$ \emph{positively satisfies} a logic program $P$ iff for each rule $r \in P$ we have that: If $\bodyp{r} \subseteq \head{R}$, then $r \in R$. We will use $\minpos{P}$ to denote the minimal (w.r.t. $\subseteq$) set of rules that positively satisfies $P$. It contains all the rules from $P$ that can be applied in the iterative manner: we apply a rule which positive body is derived by the rules applied before.

\begin{example}
Consider the following program $P$:
\begin{tabbing}
\laarule{r_1}{a}{}\\
\lbrule{r_2}{b}{a}\\
\lbrule{r_3}{d}{c}
\end{tabbing}
We have that $R_1 = \{r_1, r_2\}$ and $R_2 = \{r_1, r_2, r_3\}$ positively satisfy $P$. On the other hand $R_3 =\{r_1\}$ doe not positively satisfy $P$ as $\bodyp{r_2} \subseteq \head{R_3}$ and $r_2 \not \in R_3$.

We also have that $\minpos{P} = R_1$.
\end{example}


The \emph{reduct $P^R$} of a logic program $P$ w.r.t. a set of rules $R \subseteq P$ is obtained from $P$ by removing each rule $r$ with $\head{R} \cap \bodym{r} \neq \emptyset$.

A set of rules $R \subseteq P$ is a \emph{generating set} of a logic program $P$ iff $R = \minpos{\reduct{P}{R}}$.

\begin{definition}[Answer set]
\label{definition:as_alt}
A consistent set of literals $S$ is an \emph{answer set} of a logic program $P$ iff there is a generating set $R$ such that $\head{R} = S$.
\end{definition}

\begin{example}
Consider the following program $P$
\begin{tabbing}
\laarule{r_1}{a}{\dnot b}\\
\lbrule{r_2}{c}{d, \dnot b}\\
\lbrule{r_3}{b}{\dnot a}
\end{tabbing}
Let $R = \{r_1\}$. When constructing $\reduct{P}{R}$ we remove $r_3$ as $\bodym{r_3} \cap \head{R} \neq \emptyset$. We get that $\reduct{P}{R} = \{r_1, r_2\}$, and $\minpos{\reduct{P}{R}} = \{r_1\}$.  The rule $r_2$ is not included as $d \in \bodyp{r_2}$ cannot be derived. We have that $\minpos{\reduct{P}{R}} = R$. Therefore $R$ is a generating set of $P$ and $\{a\} = \head{R}$ is an answer set of $P$.
\end{example}

It holds that: if a set of rules $R$ is a generating set of a logic program $P$, and $S = \head{R}$ is consistent, then $R = \agras{S}{P}$.

\section{Conflicts}

Informally, two rules are conflicting, if their applicability is mutually exclusive: if the application of one rule causes the other rule to be inapplicable, and vice versa. We divide general conflicts into two disjunctive categories:
\begin{itemize}
\item direct conflicts, and
\item indirect conflicts.
\end{itemize}

In case of a direct conflict, application of a conflicting rule causes immediately the other rule to be inapplicable.

\begin{definition}[Directly Conflicting Rules]
\label{definition:conflict}
We say that rules $r_1$ and $r_2$ are \emph{directly conflicting} iff:
\begin{inparaenum}[(i)]
\item $r_1$ defeats $r_2$, and
\item $r_2$ defeats $r_1$.
\end{inparaenum}
\end{definition}

\begin{example}
\label{example:direct_conflict}
Consider the following program
\begin{tabbing}
\laarule{r_1}{a}{\dnot b}\\
\lbrule{r_2}{b}{\dnot a}
\end{tabbing}
The rules $r_1$ and $r_2$ are directly conflicting. If $r_1$ is used, then $r_2$ is not applicable, and vice versa.
\end{example}

In case of an indirect conflict, another, intermediate rule, has to be used. The following example illustrated the idea.

\begin{example}
Consider the following program
\begin{tabbing}
\laarule{r_1}{x}{\dnot b}\\
\lbrule{r_2}{b}{\dnot a}\\
\lbrule{r_3}{a}{x}
\end{tabbing}
Now, the rule $r_1$ is not able to make $r_2$ inapplicable on its own. The rule $r_3$ is also needed. Therefore we say that $r_1$ and $r_2$ are indirectly conflicting, and the conflict is formed via the rules $r_3$.
\end{example}

When trying to provide a formal definition of a general conflict, one has to address several difficulties.

First, an indirect conflict is not always effectual. The following example illustrates what we mean by that.

\begin{example}
Consider the following program.
\begin{tabbing}
\laarule{r_1}{x}{\dnot b}\\
\lbrule{r_2}{b}{\dnot a}\\
\lbrule{r_3}{a}{x, \dnot y}\\
\lbrule{r_4}{y}{}
\end{tabbing}
When the rule $r_2$ is used, the rule $r_1$ cannot be used. However, if we use $r_1$, the rule $r_2$ is still applicable as the rule $r_3$ that depends on $r_1$ and defeats $r_2$ is defeated by the fact $r_4$. Note that this cannot happen in the case of direct conflicts.
\end{example}

Second, we need to define that an indirect conflict is formed via rules that are somehow related to a conflicting rule.
\begin{example}
Consider the following program:
\begin{tabbing}
\laarule{r_1}{a}{\dnot b}\\
\lbrule{r_2}{x}{\dnot a}\\
\lbrule{r_3}{b}{}
\end{tabbing}
If we fail to see that $r_3$ does not depend on $r_2$, we can come to wrong conviction that $r_1$ and $r_2$ are conflicting via $r_3$ as
\begin{inparaenum}[(i)]
\item $r_1$ defeats $r_2$, and
\item $r_3$ defeats $r_1$.
\end{inparaenum}
\end{example}

Third, in general, the rules depending on a rule are conflicting, thus creating alternatives, in which the rule is/is not conflicting. The following example illustrates this.
\begin{example}
Consider the following program:
\begin{tabbing}
\laarule{r_1}{x}{\dnot c}\\
\lbrule{r_2}{a}{x, \dnot b}\\
\lbrule{r_3}{b}{x, \dnot a}\\
\lbrule{r_4}{c}{\dnot a}
\end{tabbing}
Since the rules $r_2$ and $r_3$ are directly conflicting, they cannot be used at the same time. If $r_2$ is used, $r_1$ and $r_4$ are conflicting via $r_2$. If $r_3$ is used, $r_1$ and $r_4$ are not conflicting.
\end{example}

In this paper we are going to address these issues from a different angle. Instead of defining a general conflict between two rules, we will move to sets of rules and define conflicts between sets of rules in the later sections.

\section{Approach to Direct Conflicts}

In this section we recapitulate our semantics for directs conflicts \cite{Simko:2013tm}, which we generalize in this paper for the case of general conflicts.

We say that a rule $r_1$ \emph{directly overrides} a rule $r_2$ w.r.t. a preference relation $<$ iff
\begin{inparaenum}[(i)]
\item $r_1$ and $r_2$ are directly conflicting, and
\item $r_2 < r_1$.
\end{inparaenum}

The \emph{reduct} $\reduct{\lpp{P}}{R}$ of a logic program with preferences $\lpp{P} = (P,<)$ w.r.t. a set of rules $R \subseteq P$ is obtained from $P$ by removing each rule $r_1 \in P$, for which there is a rule $r_2 \in R$ such that:
\begin{itemize}
\item $r_2$ defeats $r_1$, and
\item $r_1$ does not directly override $r_2$ w.r.t. $<$.
\end{itemize}

A set of rules $R \subseteq P$ is a \emph{preferred generating set} of a logic program with preferences $\lpp{P} = (P,<)$ iff $R = \minpos{\reduct{\lpp{P}}{R}}$.

A consistent set of literals $S$ is a \emph{preferred answer set} of a logic program with preferences $\lpp{P}$ iff there is a preferred generating set $R$ of $\lpp{P}$ such that $\head{R} = S$.

We will use $\pasd{\lpp{P}}$ to denote the set of all the preferred answer sets of $\lpp{P}$ according to this definition.

It holds that each preferred generating set of $\lpp{P} = (P,<)$ is a generating set of $P$.

\section{Principles}

An important direction in preference handling research is the study of principles that a reasonable semantics should satisfy. Brewka and Eiter have proposed first two principles \cite{Brewka:1999uv}.

Principle \ref{principle:i} tries to capture the meaning of preferences. If two answer sets are generated by the same rules except for two rules, the one generated by a less preferred rule is not preferred.

\begin{principle}[\cite{Brewka:1999uv}]
\label{principle:i}
Let $\lpp{P} = (P,<)$ be a logic program with preferences, $S_1, S_2$ be two answer sets of $P$. Let $\agras{S_1}{P} = R \cup \{ r_1 \}$ and $\agras{S_2}{P} = R \cup \{ r_2 \}$ for $R \subset P$. Let $r_2 < r_1$. Then $S_2$ is not a preferred answer set of $\lpp{P}$.
\end{principle}

Principle \ref{principle:ii} says that the preferences specified on a rule with an unsatisfied positive body are irrelevant.

\begin{principle}[\cite{Brewka:1999uv}]
\label{principle:ii}
Let $S$ be a preferred answer set of a logic program with preferences $\lpp{P} = (P,<)$, and $r$ be a rule such that $\bodyp{r} \not \subseteq S$. Then $S$ is a preferred answer set of a logic program with preferences $\lpp{P}^\prime = (P^\prime,<^\prime)$, where $P^\prime = P \cup \{ r\}$ and  $<^\prime \cap (P \times P) = <$.
\end{principle}

Principle \ref{principle:iii}\footnote{It is an idea from Proposition 6.1 from \cite{Brewka:1999uv}. Brewka and Eiter did not consider it as a principle. On the other hand \cite{Sefranek:2008wo} did.} requires that a program has a preferred answer set whenever a standard answer set of the underlying program exists. It follows the view that the addition of preferences should not cause a consistent program to be inconsistent.

\begin{principle}
\label{principle:iii}
Let $\lpp{P} = (P,<)$ be a logic program with preferences. If $P$ has an answer set, then $\lpp{P}$ has a preferred answer set.
\end{principle}

Before we proceed, we remind that our approach to preference handling is for general conflicts, and understands preferences on rules as a mechanism for expressing exception between rules. Using this view, we show that Principle \ref{principle:ii} and Principle \ref{principle:iii} should be violated by a semantics, and hence are not relevant under this understanding of preferences.

\begin{example}
Consider the following program $\lpp{P} = (P,<)$
\begin{tabbing}
\laarule{r_1}{select(a)}{\dnot \neg select(a)}\\
\lbrule{r_2}{select(b)}{\dnot \neg select(b)}\\
\\
\lbrule{r_3}{\neg select(a)}{select(b)}\\
\\
$r_2 < r_1$
\end{tabbing}
The program is stratified, and has the unique answer set $S = \{\neg select(a), select(b)\}$. Since there are no conflicts between the rules, the unique answer set should be preferred.

We construct $\lpp{P}^\prime = (P^\prime, <)$, $P^\prime = P \cup \{ r_4 \}$, by adding the rule
\begin{tabbing}
\laarule{r_4}{\neg select(b)}{select(a)}
\end{tabbing}
We have an indirect conflict between the rules $r_1$ and $r_2$ via $r_3$ and $r_4$. The rule $r_1$ being preferred, $S$ should not be a preferred answer set of $\lpp{P}^\prime$.

Hence Principle \ref{principle:ii} is violated: $\bodyp{r_4} =\{ select(a) \} \not \subseteq S$, but $S$ is not a preferred answer set of $\lpp{P}^\prime$.
\end{example}

\begin{example}
Consider the following program $\lpp{P} = (P,<)$.
\begin{tabbing}
\laarule{r_1}{select(a)}{\dnot \neg select(a)}\\
\lbrule{r_2}{\neg select(a)}{\dnot select(a)}\\
\\
$r_2 < r_1$
\end{tabbing}
When we interpret preference $r_1 < r_2$ as a way of saying that $r_1$ defines an exception to $r_2$ and not vice versa, the program has the following meaning:
\begin{tabbing}
\laarule{r_1}{select(a)}{}\\
\lbrule{r_2}{\neg select(a)}{\dnot select(a)}
\end{tabbing}
Hence $S = \{select(a)\}$ is the unique preferred answer set of $\lpp{P}$.

We construct $\lpp{P}^\prime = (P^\prime,<)$, $P^\prime = P \cup \{r_3\}$, by adding the rule
\begin{tabbing}
$r_3: \lrule{inc}{select(a), \dnot inc}$
\end{tabbing}
The program $\lpp{P}^\prime$ has the following meaning:
\begin{tabbing}
\laarule{r_1}{select(a)}{}\\
\lbrule{r_2}{\neg select(a)}{\dnot select(a)}\\
\lbrule{r_3}{inc}{select(a), \dnot inc}
\end{tabbing}
The program has no answer set, and hence $\lpp{P}^\prime$ has no preferred answer set.

Hence Principle \ref{principle:iii} is violated: The program $P^\prime$ has an answer set, but $\lpp{P}^\prime$ has no preferred answer set.
\end{example}

\section{Approach One to General Conflicts}

In this section we generalize our approach to direct conflicts to the case of general conflicts. As we have already noted, we deliberately avoid defining what a general conflict between two rules is. We will define when two sets of rules are conflicting instead. For this reason we develop an alternative definition of an answer set as a set of sets of rules, upon which the semantics for preferred answer sets will be defined.

\subsection{Alternative Definition of Answer Sets}

A building block of the alternative definition of answer sets is a fragment. The intuition behind a fragment is that it is a set of rules that can form the one hand side of a conflict. The positive bodies of the rules must be supported in a non-cyclic way.

\begin{definition}[Fragment]
A set of rules $R \subseteq P$ is a \emph{fragment} of a logic program $P$ iff $\minpos{R} = R$.
\end{definition}

\begin{example}
\label{alt_as:example:running:1}
Consider the following program $P$ that we will use to illustrate the definitions of this paper.
\begin{tabbing}
\laarule{r_1}{a}{x}\\
\lbrule{r_2}{x}{\dnot b}\\
\lbrule{r_3}{b}{\dnot a}
\end{tabbing}
The sets $F_1 = \emptyset$, $F_2 = \{r_2\}$, $F_3 = \{r_3\}$, $F_4=\{r_2, r_1\}$, $F_5=\{r_2, r_3\}$, $F_6 = \{r_1, r_2, r_3\}$ are all the fragments of the program. For example, $\{r_1\}$ is not a fragment as $\minpos{\{r_1\}} = \emptyset$.
\end{example}

\begin{notation}
We will denote by $\frag{P}$ the set of all the fragments of a program $P$.
\end{notation}

\begin{notation}
Let $P$ be a logic program and $E \subseteq \frag{P}$.

We will denote $\exrules{E} = \bigcup_{X \in E} X$, and $\head{E} = \head{\exrules{E}}$.
\end{notation}

Given a guess of fragments, we define the reduct. Since fragments are sets of rules, we can speak about defeating between fragments.

\begin{definition}[Reduct]
Let $P$ be a logic program and $E \subseteq \frag{P}$.

The reduct $\reduct{P}{E}$ of $P$ w.r.t. $E$ is obtained from $\frag{P}$ by removing each fragment $X \in \frag{P}$ for which there is $Y \in E$ that defeats $X$.
\end{definition}

\begin{example}[Example \ref{alt_as:example:running:1} continued]
\label{alt_as:example:running:2}
Let $E_1 = \{ F_1, F_2, F_4\}$. We have that $\reduct{P}{E_1} = \{F_1, F_2, F_4\}$. The fragments $F_3$, $F_5$, and $F_6$ are removed as they contain the rule $r_3$ which is defeated by $F_4 \in E_1$.

Let $E_2 = \{ F_2 \}$. We have that $\reduct{P}{E_2} = \{ F_1, F_2, F_3, F_4, F_5, F_6 \}$. Since no rule has $x$ in its negative body, no fragment is removed.
\end{example}

A stable fragment set, an alternative notion to the notion of answer set, is a set of fragments that is stable w.r.t. to the reduction.

\begin{definition}[Stable fragment set]
A set $E \subseteq \frag{P}$ is a \emph{stable fragment set} of a program $P$ iff $\reduct{P}{E} = E$.
\end{definition}

\begin{example}[Example \ref{alt_as:example:running:2} continued]
\label{alt_as:example:running:3}
We have that $\reduct{P}{E_1} = E_1$, so $E_1$ is a stable fragment set. On the other hand, $E_2$ is not a stable fragment set as $\reduct{P}{E_2}  \neq E_2$.
\end{example}

\begin{proposition}
\label{pas_gen:proposition:ext_eq_gen}
Let $P$ be a logic program, and $E \subseteq \frag{P}$.

$E$ is a stable fragment set of $P$ iff $\exrules{E}$ is a generating set of $P$ and $E = \{ T : T = \minpos{T} \mbox{ and } T \subseteq \exrules{E}\}$.
\end{proposition}

From Proposition \ref{pas_gen:proposition:ext_eq_gen} we directly have that the following is an alternative definition of answer sets.

\begin{proposition}
Let $P$ be a logic program and $S$ a consistent set of literals.

$S$ is an answer set of $P$ iff there is a stable fragment set $E$ of $P$ such that $\head{E} = S$.
\end{proposition}

\begin{example}[Example \ref{alt_as:example:running:3} continued]
$E_1 = \{ F_1, F_2, F_4 \}$ and $E_3 = \{ F_1, F_3\}$ are the only stable fragment sets of the program. The sets $\{ a, x\} = \head{E_1}$ and $\{b\} = \head{E_3}$ are the only answer sets of the program.
\end{example}

\subsection{Preferred Answer Sets}

In this subsection we develop our first definition of preferred answer sets for general conflicts from the alternative definition of answer sets based on stable fragment sets.

The basic intuition behind the approach is that \emph{a rule cannot be defeated by a less preferred conflicting rule}. This intuition is realized by modifying the definition of reduct. We do not allow a fragment $X$ to be removed because of a fragment $Y$ if $Y$ uses less preferred conflicting rules. For this purpose we use the term ``override''.

\begin{definition}[Conflicting Fragments]
Fragments $X$ and $Y$ are \emph{conflicting} iff
\begin{inparaenum}[(i)]
\item $X$ defeats $Y$, and
\item $Y$ defeats $X$.
\end{inparaenum}
\end{definition}

\begin{example}[Example \ref{alt_as:example:running:1} continued]
\label{pas_ind_gen:example:running:1}
Let us recall the fragments:  $F_2 = \{r_2\}$, $F_3 = \{r_3\}$, and $F_4 =\{r_2, r_1\}$. The fragments $F_3$ and $F_4$ are conflicting as $\head{r_3} \in \bodym{r_2}$ and $\head{r_1} \in \bodym{r_3}$. On the other hand, $F_2$ and $F_3$ are not conflicting. The fragment $F_3$ defeats $F_2$, but not the other way around as $\head{r_2} \not \in \bodym{r_3}$.
\end{example}

\begin{definition}[Override]
Let $X$ and $Y$ be conflicting fragments. We say that $X$ \emph{overrides} $Y$ w.r.t. a preference relation $<$ iff
for each $r _1 \in X$ that is defeated by $Y$, there is $r_2 \in Y$ defeated by $X$, and $r_2 < r_1$.
\end{definition}

\begin{example}[Example \ref{pas_ind_gen:example:running:1} continued]
\label{pas_ind_gen:example:running:2}
Let us continue with preference $r_2 < r_3$. We have that $F_3$ overrides $F_4$ and $F_3$ overrides $F_6$. On the other hand $F_3$ does not override $F_2$ because $F_2$ does not defeat $F_3$. From the following Proposition \ref{pas_gen:proposition:override_asym} we also have that $F_6$ does not override $F_6$.
\end{example}

\begin{proposition}
\label{pas_gen:proposition:override_asym}
Let $\lpp{P} = (P,<)$ be a logic program with preferences, $X$ and $Y$ be fragments of $P$.

If $X$ overrides $Y$ w.r.t. $<$, then $Y$ does not override $X$ w.r.t. $<$.
\end{proposition}


When constructing the reduct w.r.t. a guess, a fragment $X$ cannot be removed because of a fragment $Y$ which is overridden by $X$.

\begin{definition}[Reduct]
Let $\lpp{P} = (P,<)$ be a logic program with preferences, and $E \subseteq \frag{P}$.

The reduct $\reduct{\lpp{P}}{E}$ of $\lpp{P}$ w.r.t. $E$ is obtained from $\frag{P}$ by removing each $X \in \frag{P}$ such that there is $Y \in E$ that:
\begin{itemize}
\item $Y$ defeats $X$, and
\item $X$ does not override $Y$ w.r.t. $<$.
\end{itemize}
\end{definition}

\begin{example}[Example \ref{pas_ind_gen:example:running:2} continued]
\label{pas_ind_gen:example_run_3}
Let $E_1 = \{ F_1, F_2, F_4\}$. We have that $\reduct{\lpp{P}}{E_1} = \{F_1, F_2, F_3, F_4\}$. Now, the fragment $F_3$ is not removed as the only fragment from $E_1$ that defeats it is $F_4$, but $F_3$ overrides $F_4$.
\end{example}

\begin{definition}[Preferred stable fragment set]
Let $\lpp{P} = (P,<)$ be a logic program with preferences., and $E \subseteq \frag{P}$.

We say that $E$ is a \emph{preferred stable fragment set} of $\lpp{P}$ iff $\reduct{\lpp{P}}{E} = E$.
\end{definition}

\begin{example}[Example \ref{pas_ind_gen:example:running:2} continued]
\label{pas_ind_gen:example:running:3}
Now we have that $\reduct{\lpp{P}}{E_1} \neq E_1$, so $E_1$ is not a preferred stable fragment set. On the other hand, $E_3 = \{ F_1, F_3\}$ is a preferred stable fragment set as $\reduct{\lpp{P}}{E_3}  = E_3$.
\end{example}

\begin{definition}[Preferred answer set]
\label{pas_gen:definition:pas}
Let $\lpp{P} = (P,<)$ be a logic program with preferences, and $S$ be a consistent set of literals.

$S$ is a \emph{preferred answer set} of $\lpp{P}$ iff there is a preferred stable fragment set $E$ of $\lpp{P}$ such that $\head{E} = S$.

We will use $\pasindgen{\lpp{P}}$ to denote the set of all the preferred answer sets of $\lpp{P}$ according to this definition.
\end{definition}

\begin{example}[Example \ref{pas_ind_gen:example:running:3} continued]
\label{pas_ind_gen:example_run_4}
The set $E_3 = \{F_1, F_3\}$ is the only preferred stable fragment set, and $\{b\} = \head{E_3}$ is the only preferred answer set of the program.
\end{example}

\begin{proposition}
\label{pas_gen:proposition:pref_ext_is_ext}
Let $\lpp{P} = (P,<)$ be a logic program with preferences, and $E \subseteq \frag{P}$.

If $E$ is a preferred stable fragment set of $\lpp{P}$, then $E$ is a stable fragment set of $P$.
\end{proposition}

\subsection{Properties}

Preferred answer sets as defined in Definition \ref{pas_gen:definition:pas} enjoy following nice properties.

\begin{proposition}
Let $\lpp{P} = (P,<)$ be a logic program with preferences. Then $\pasindgen{\lpp{P}} \subseteq \as{P}$.
\end{proposition}

\begin{proposition}
Let $\lpp{P} = (P,\emptyset)$ be a logic program with preferences. Then $\pasindgen{\lpp{P}} = \as{P}$.
\end{proposition}

\begin{proposition}
Preferred answer sets as defined in Definition \ref{pas_gen:definition:pas} satisfy Principle \ref{principle:i}.
\end{proposition}

\begin{proposition}
Let $\lpp{P}_1 = (P, <_1)$, $\lpp{P}_2 = (P, <_2)$ be logic programs with preferences such that $<_1 \subseteq <_2$.

Then $\pasindgen{\lpp{P}_2} \subseteq \pasindgen{\lpp{P}_1}$.
\end{proposition}

On the subclass of stratified programs, the semantics is equivalent to the answer set semantics. We consider this property to be an important one as stratified programs contain no conflicts.

\begin{proposition}
\label{pas_ind_gen:proposition:strat}
Let $\lpp{P} = (P,<)$ be a logic program with preferences such that $P$ is stratified. Then $\pasindgen{\lpp{P}} = \as{P}$.
\end{proposition}

The following example illustrates how the approach works on stratified programs.

\begin{example}
\label{example:pasi_strat}
Consider a problematic program from \cite{Brewka:1999uv}:
\begin{tabbing}
\laarule{r_1}{a}{\dnot b}\\
\lbrule{r_2}{b}{}\\
\\
\tgo $r_2 < r_1$
\end{tabbing}

The program is stratified and has a unique answer set $S = \{b\}$.

The program has the following fragments $F_0 = \emptyset$, $F_1 = \{r_1\}$, $F_2 = \{ r_2 \}$, $F_3 = \{ r_1, r_2 \}$. The set $E = \{ F_0, F_2\}$ is a unique stable fragment set.

We have that $F_2$ defeats both $F_1$, and $F_3$. Neither $F_1$ nor $F_3$ override $F_2$ as they are not conflicting with $F_2$. This is the reason why preference $r_2 < r_1$ is ignored here, and both $F_1$ and $F_3$ are removed during the reduction: $\reduct{\lpp{P}}{E} = \{F_0, F_2\} = E$. Therefore $S$ is a unique preferred answer set.
\end{example}

From the computational complexity point of view, so far, we have established only the upper bound. Establishing the lower bound remains among open problems for future work.

\begin{proposition}
Given a logic program with preferences $\lpp{P}$, deciding whether $\lpp{P}$ has a preferred answer set  is in $\Sigma_3^P$.
\end{proposition}

\section{Approach Two to General Conflicts}

If we have an application domain, where we can relax the requirements for preference handling in a sense that we no longer require preferences between non-conflicting rules to be ignored, we can ensure that the semantics stays in the NP complexity class.

In this section we simplify our first approach by using the following intuition for preference handling: \emph{a rule cannot be defeated by a less preferred rule or a rule depending on a less preferred rule}.

The definition of the approach follows the structure of our approach for direct conflicts. The presented intuition is realized using a set $\trules{r}{R}$ in the definition of reduct.

\begin{definition}[Reduct]
\label{pasindirect:definition:reduct}
Let $\mathcal{P} = (P,<)$ be a logic program with preferences, and $R \subseteq P$ be a set of rules.

The {\em reduct} $\mathcal{P}^R$ of $\lpp{P}$ w.r.t. $R$ is obtained from $P$ by removing each rule $r \in P$ such that $\bodym{r} \cap \head{\trules{r}{R}} \neq \emptyset$, where
$\trules{r}{R} = \minpos{\{ p \in R : p \not < r\}}$.
\end{definition}

\begin{example}[Example \ref{pas_ind_gen:example:running:2} continued]
\label{pas_ind:example:running:1}
Let us recall the program:
\begin{tabbing}
\laarule{r_1}{a}{x}\\
\lbrule{r_2}{x}{\dnot b}\\
\lbrule{r_3}{b}{\dnot a}\\
\\
$r_2 < r_3$
\end{tabbing}
Let $R_1 = \{ r_1, r_2 \}$. We have that $\trules{r_1}{R_1} = R_1$, $\trules{r_2}{R_1} = R_1$. On the other hand $\trules{r_3}{R_1} = \emptyset$ as $r_2 < r_3$ and $r_1$ depends on $r_2$. No rule less preferred, and no rule that depends on a rule less preferred than $r_3$ can be used to defeat $r_3$. In this case no rule can defeat $r_3$.

Hence $\reduct{\lpp{P}}{R_1} = \{ r_1, r_2, r_3 \}$.
\end{example}

\begin{definition}[Preferred generating set]
Let $\lpp{P} = (P,<)$ be a logic program with preferences, and $R$ be a generating set of $P$.

We say that $R$ is a \emph{preferred generating set} of $\lpp{P}$ iff $R = \minpos{\reduct{\lpp{P}}{R}}$.
\end{definition}

\begin{example}[Example \ref{pas_ind:example:running:1} continued]
\label{pas_ind:example:running:2}
We have that $\minpos{\reduct{\lpp{P}}{R_1}} = P \neq R_1$. Hence $R_1$ is not a preferred generating set.
\end{example}

\begin{definition}[Preferred answer set]
\label{pasindirect:definition:pas}
Let $\lpp{P} = (P,<)$ be a logic program with preferences, and $S$ be a consistent set of literals.

$S$ is a preferred answer set of $\lpp{P}$ iff there is a preferred generating set $R$ such that $S = \head{R}$.

We will use $\pasindno{\lpp{P}}$ to denote the set of all the preferred answer sets of $\lpp{P}$ according to this definition.
\end{definition}

\begin{example}[Example \ref{pas_ind:example:running:2} continued]
\label{pas_ind:example:running:3}
The set $R_2 = \{r_3\}$ is the only preferred generating set, and $\{b\} = \head{R_2}$ is the only preferred answer set.
\end{example}

\subsection{Transformation}

It turns out that the second approach can be characterized by a transformation from programs with preferences to programs without preferences in a way that the answer sets of the transformed  program correspond (modulo new special-purpose literals) to the preferred answer sets of an original program.

The idea of the transformation is to use special-purpose literals and auxiliary rules in order to allow a rule $r$ to be defeated only by $\trules{r}{R}$ where $R$ is a preferred generating set guess. We first present the definition of the transformation and then explain each rule.

\begin{notation}
If $r$ is a rule of a program $P$, then $n_r$ denotes a new literal not occurring in $P$.

If $r$ is a rule of a program $P$, and $x$ is a literal of $P$, then $x^r$ denotes a new literal not occurring in $P$ and different from $n_q$ for each $q \in P$. For a set of literals $S$, $S^r$ denotes $\{ x^r : x \in S \}$.

We will also use $inc$ to denote a literal not occurring in $P$ and different from all previously mentioned literals.
\end{notation}

\begin{definition}[Transformation]
Let $\lpp{P} = (P,<)$ be a logic program with preferences.

Let $r$ be a rule. Then $\transr{r}{\lpp{P}}$ is the set of the rules
\begin{eqnarray}
\label{pas_ind:rule:head_r}\alrule{\head{r}}{n_r}\\
\label{pas_ind:rule:n_r}\alrule{n_r}{\bodyp{r}, \dnot \bodym{r}^r}
\end{eqnarray}
and the rule
\begin{eqnarray}
\label{pas_ind:rule:head_p_r}\alrule{\head{p}^r}{\bodyp{p}^r, n_p}
\end{eqnarray}
for each $p \in P$ such that $p \not < r$, and the rule
\begin{eqnarray}
\label{pas_ind:rule:const}\alrule{inc}{n_r, x, \dnot inc}
\end{eqnarray}
for each $x \in \bodym{r}$.

$\trans{\lpp{P}} = \bigcup_{r \in P} t_{\mathcal{P}}(r)$.
\end{definition}

A preferred generating set guess $R$ is encoded using $n_r$ literals. The meaning of a literal $n_r$ is that a rule $r$ was applied. In order to derive $n_r$ literals, we split each rule $r$ of a program into two rules: The rule \eqref{pas_ind:rule:n_r} derives literal $n_r$, and the rule \eqref{pas_ind:rule:head_r} derives the head of the original rule $r$.

The special-purpose literals $x^r$ are used in the negative body of the rule \eqref{pas_ind:rule:n_r} in order to ensure that only $\trules{r}{R}$ can defeat a rule $r$. The $x^r$ literals are derived using the rules of the form \eqref{pas_ind:rule:head_p_r}.

The rules of the form \eqref{pas_ind:rule:const} ensure that no answer set of $\trans{\lpp{P}}$ contains both $n_r$ and $x$. This condition is needed in order to ensure that $R$ is also a generating set.

\begin{example}
\label{pas_ind:example:tr_1}
Consider again our running program $\lpp{P}$:
\begin{tabbing}
\laarule{r_1}{a}{x}\\
\lbrule{r_2}{x}{\dnot b}\\
\lbrule{r_3}{b}{\dnot a}\\
\\
\tgo $r_2 < r_3$\\
\end{tabbing}
$\trans{\lpp{P}}$ is as follows:\\
\begin{tabbing}
\larule{}{a}{\rname{r_1}}{}{0.8em}{0em}
\tset{2em}
\larule{}{x}{\rname{r_2}}{}{0.8em}{0em}
\tset{3em}
\larule{}{b}{\rname{r_3}}{}{0.8em}{0em}\\
\lbrule{}{\rname{r_1}}{x}
\tgo
\lbrule{}{\rname{r_2}}{\dnot \rlit{b}{r_2}}
\tgo
\lbrule{}{\rname{r_3}}{\dnot \rlit{a}{r_3}}\\
\\
\lbrule{}{\rlit{a}{r_1}}{\rlit{x}{r_1}, \rname{r_1}}
\tgo
\lbrule{}{\rlit{a}{r_2}}{\rlit{x}{r_2}, \rname{r_1}}
\tgo
\lbrule{}{\rlit{a}{r_3}}{\rlit{x}{r_3}, \rname{r_1}}\\
\lbrule{}{\rlit{x}{r_1}}{\rname{r_2}}
\tgo
\lbrule{}{\rlit{x}{r_2}}{\rname{r_2}}\\
\lbrule{}{\rlit{b}{r_1}}{\rname{r_3}}
\tgo
\lbrule{}{\rlit{b}{r_2}}{\rname{r_3}}
\tgo
\lbrule{}{\rlit{b}{r_3}}{\rname{r_3}}\\
\\
\lbrule{}{inc}{\rname{r_2}, b, \dnot inc}\\
\lbrule{}{inc}{\rname{r_3}, a, \dnot inc}
\end{tabbing}
Now, as $r_2 < r_3$, a transformed rule deriving $\rlit{x}{r_3}$ coming from $r_2$ is not included.
\end{example}

The transformation captures the semantics of preferred answer sets as defined in Definition \ref{pasindirect:definition:pas}.

\begin{proposition}
\label{proposition:tr_eq}
Let $\lpp{P} = (P,<)$ be a logic program with preferences. Let $\Lit$ be a set of all the literals constructed from the atoms of $P$, and $\namess{S}{\lpp{P}} = \{ n_r : r \in \agras{S}{P}\}$, and $Aux(S) = \bigcup_{r\in P} \head{\trules{r}{R}}^r$, where $R = \agras{S}{P}$.

If $S$ is a preferred answer set of $\lpp{P}$, then $A = S \cup \namess{S}{\lpp{P}} \cup Aux(S)$ is an answer set of $\trans{\lpp{P}}$.

If $A$ is an answer set of $\trans{\lpp{P}}$, then $S = A \cap \Lit$ is a preferred answer set of $\lpp{P}$, and $A = S \cup \namess{S}{\lpp{P}} \cup Aux(S)$.
\end{proposition}

\subsection{Properties}

Preferred answer sets as defined in Definition \ref{pasindirect:definition:pas} enjoy several nice properties.

\begin{proposition}
Let $\lpp{P} = (P,<)$ be a logic program with preferences. Then $\pasindno{\lpp{P}} \subseteq \as{P}$.
\end{proposition}

\begin{proposition}
\label{proposition:empty_pref}
Let $\lpp{P} = (P,\emptyset)$ be a logic program with preferences. Then $\pasindno{\lpp{P}} = \as{P}$.
\end{proposition}

\begin{proposition}
Preferred answer sets as defined in Definition \ref{pasindirect:definition:pas} satisfy Principle \ref{principle:i}.
\end{proposition}

\begin{proposition}
Let $\lpp{P}_1 = (P, <_1)$ and $\lpp{P}_2 = (P, <_2)$ be logic programs with preferences such that $<_1 \subseteq <_2$. Then $\pasindno{\lpp{P}_2} \subseteq \pasindno{\lpp{P}_1}$.
\end{proposition}

The approach two is not equivalent to the answer set semantics for the subclass of stratified programs.

\begin{proposition}
There is a logic program with preferences $\lpp{P} = (P,<)$ where $P$ is stratified and $\pasindno{\lpp{P}} = \emptyset$.
\end{proposition}

Example \ref{example:pasino_strat} shows such a program. Example \ref{example:pasi_strat} and \ref{example:pasino_strat} illustrate the main difference between the two approaches. While $\snpasindgen$ ignores preferences between non-conflicting rules, $\snpasindno$ is not always able to do so.

\begin{example}
\label{example:pasino_strat}
Consider again the program from Example \ref{example:pasi_strat}:
\begin{tabbing}
\laarule{r_1}{a}{\dnot b}\\
\lbrule{r_2}{b}{}\\
\\
\tgo $r_2 < r_1$
\end{tabbing}

The program is stratified and has a unique answer set $S = \{b\}$. A unique generating set $R = \{ r_2 \}$ corresponds to the answer set $S$.

We have that $\trules{r_1}{R} = \emptyset$. The rule $r_2$ is not included as $r_2 < r_1$. Due to a simplicity of the approach, preference $r_2 < r_1$ is not ignored. Hence $\head{\trules{r_1}{R}} \cap \bodym{r_1} = \emptyset$, and $r_1 \in \reduct{\lpp{P}}{R}$. From that $\minpos{\reduct{\lpp{P}}{R}} \neq R$, and $S$ is not a preferred answer set.
\end{example}

On the other hand the approach stays in the NP complexity class.

\begin{proposition}
Deciding whether $\pasindno{\lpp{P}} \neq \emptyset$ for a logic program with preferences $\lpp{P}$ is NP-complete.
\end{proposition}

\begin{proof}
Membership: Using Proposition \ref{proposition:tr_eq}, we can reduce the decision problem $\pasindno{\lpp{P}} \neq \emptyset$ to the problem $\as{\trans{\lpp{P}}} \neq \emptyset$ (in polynomial time), which is in NP. Hardness: Deciding $\as{P} \neq \emptyset$ for a program $P$ is NP-complete. Using Proposition \ref{proposition:empty_pref} we can reduce it to the decision $\pasindno{(P,\emptyset)} \neq \emptyset$.
\end{proof}

\section{Relation between the Approaches of this Paper}

It turns out that the approaches of this paper form a hierarchy, which does not collapse.

\begin{notation}
Let $A$ and $B$ be names of semantics.

We write $A \subseteq B$ iff each preferred answer set according to $A$ is a preferred answer set according to $B$.

We write $A = B$ iff $A \subseteq B$ and $B \subseteq A$.
\end{notation}

\begin{proposition}
$\snpasindno \subseteq \snpasindgen \subseteq \snpasd$
\end{proposition}

\begin{proposition}
$\snpasd \not \subseteq \snpasindgen$
\end{proposition}

\begin{proposition}
$\snpasindgen \not \subseteq \snpasindno$
\end{proposition}

We interpret the results as follows. The semantics $\snpasd$ is the reference semantics for the case of direct conflicts. The semantics $\snpasindno$ and $\snpasindgen$ extend the semantics to the case of indirect conflicts. The semantics $\snpasindgen$ ignores preferences between non-conflicting rules, e.g. it is equivalent to the answer set semantics for the subclass of stratified programs (Stratified programs contain no conflicts). If an application domain allows it, we can drop the requirement for ignoring preferences between non-conflicting rules and use the semantics $\snpasindno$ that stays in the NP complexity class. The semantics $\snpasindno$ is sound w.r.t. $\snpasindgen$ but it is not complete w.r.t. $\snpasindgen$. Some preferred answer sets according to $\snpasindgen$ are not preferred according to $\snpasindno$ due to preferences between non-conflicting rules.

\section{Relation to Existing Approaches}

Schaub and Wang \cite{Schaub:2003uf} have shown that the approaches \cite{Delgrande:2002uc,Wang:2000tq,Brewka:1999uv}, referred here as $\snpasdst$, $\snpaswzl$, $\snpasbe$ form a hierarchy.

\begin{proposition}[\cite{Schaub:2003uf}]
$\snpasdst \subseteq \snpaswzl \subseteq \snpasbe$
\end{proposition}

We have shown that our approach for direct conflicts continues in this hierarchy \cite{Simko:2013tm}.

\begin{proposition}[\cite{Simko:2013tm}]
$\snpasbe \subseteq \snpasd$
\end{proposition}

The relations $\snpasdst \subseteq \snpasindno$ and $\snpaswzl \subseteq \snpasindgen$ are the only subset relation between our semantics for general conflicts $\snpasindno$, $\snpasindgen$ and $\snpasdst$, $\snpaswzl$ and $\snpasbe$.


\begin{proposition}
$\snpasdst \subseteq \snpasindno$.
\end{proposition}

\begin{proposition}
$\snpaswzl \subseteq \snpasindgen$.
\end{proposition}

\begin{proposition}
$\snpasindno \not \subseteq \snpasbe$.
\end{proposition}

\begin{corollary}~
\begin{itemize}
\item $\snpasindno \not \subseteq \snpaswzl$, $\snpasindno \not \subseteq \snpasdst$,
\item $\snpasindgen \not \subseteq \snpasbe$, $\snpasindgen \not \subseteq \snpaswzl$, $\snpasindgen \not \subseteq \snpasdst$.
\end{itemize}
\end{corollary}

\begin{proposition}
$\snpaswzl \not \subseteq \snpasindno$
\end{proposition}

The overall hierarchy of the approaches is depicted in Figure \ref{figure:hierarchy}.

\begin{figure}[h]
\caption{The hierarchy of the approaches.}
\label{figure:hierarchy}
\begin{center}
\begin{tikzpicture}
	\coordinate [label=$\snpasindno$] (INO) at (2,2);
	\coordinate [label=$\subseteq$] (S1) at (3,2);
	\coordinate [label=$\snpasindgen$] (I) at (4,2);

	\coordinate [label=$\snpasdst$] (DST) at (0,1);
	\coordinate [label=$\subseteq$] (S3) at (1,0.5);
	\coordinate [label=$\subseteq$] (S4) at (1,1.5);
	\coordinate [label=$\subseteq$] (S5) at (5,0.5);
	\coordinate [label=$\subseteq$] (S6) at (5,1.5);
	\coordinate [label=$\snpasd$] (D) at (6,1);
	
	\draw (2.7,1.3) node[rotate=30]  {$\subseteq$};
		
	\coordinate [label=$\snpaswzl$] (WZL) at (2,0);
	\coordinate [label=$\subseteq$] (S2) at (3,0);	
	\coordinate [label=$\snpasbe$] (BE) at (4,0);	
\end{tikzpicture}
\end{center}
\end{figure}
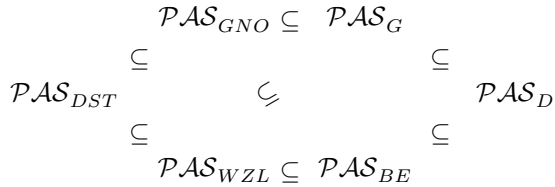

\section{An Example}

In this section we show that the approaches of this paper handle correctly  the program of Example \ref{example:motivation} from Introduction.  We remind that neither of the approaches $\snpasdst$, $\snpaswzl$ and $\snpasbe$ provides intended preferred answer sets.

\begin{example}
\label{example:my}
We recall the program:
\begin{tabbing}
\laarule{r_1}{nice(car_1)}{}\\
\lbrule{r_2}{safe(car_2)}{}\\
\\
\lbrule{r_3}{rec(car_1)}{nice(car_1), \dnot \neg rec(car_1)}\\
\lbrule{r_4}{rec(car_2)}{nice(car_2), \dnot \neg rec(car_2)}\\
\\
\lbrule{u_1}{\neg rec(car_2)}{rec(car_1)}\\
\lbrule{u_2}{\neg rec(car_1)}{rec(car_2)}\\
\\
\lbrule{u_3}{rec(car_1)}{safe(car_1), \dnot \neg rec(car_1)}\\
\lbrule{u_4}{rec(car_2)}{safe(car_2), \dnot \neg rec(car_2)}\\
\\
\tgo $r_i < u_j$ for each $i$ and $j$.
\end{tabbing}
The program has two answer sets $S_1 = \{ rec(car_1), \neg rec(car_2) \} \cup F$ and $S_2 = \{ \neg rec(car_1), rec(car_2) \} \cup F$ where $F = \{ nice(car_1), safe(car_2) \}$. As we mentioned in Introduction, $S_2$ is the intended unique preferred answer set.

\paragraph{$\snpasindgen$}:
We start by listing fragments of the program. We denote by $F_i$ fragments formed by the facts. Let $F_0 = \emptyset$, $F_1 = \{r_1\}$, $F_2 = \{r_2\}$, $F_3 = \{ r_1, r_2\}$.

The rules $r_3$ and $u_4$ are conflicting. We denote by $A_i$ fragments containing the rule $r_3$: $A_1 = \{ r_1, r_3\}$, $A_2 = \{ r_1, r_3, u_1\}$,  $A_3 = \{r_1, r_2, r_3\}$, $A_4 = \{ r_1, r_2, r_3, u_1 \}$.

We denote by $B_i$ fragments containing the rule $u_4$. Let $B_1 = \{ r_2, u_4\}$, $B_2 =\{ r_2, u_4, u_2 \}$, $B_3 = \{r_1, r_2, u_4\}$, $B_4 = \{r_1, r_2, u_4, u_2\}$.

A stable fragment set $E_1 = \{ F_0, F_1, F_2, F_3, A_1, A_2, A_3, A_4 \}$ corresponds to the answer set $S_1$ and a stable fragment set $E_2 = \{F_0, F_1, F_2, F_3, B_1, B_2, B_3, B_4\}$ corresponds to the answer set $S_2$.

We have that $B_3$ overrides both $A_2$ and $A_4$. Hence $B_3 \in \reduct{\lpp{P}}{E_1}$, and $\reduct{\lpp{P}}{E_1} \neq E_1$. Hence $S_1$ is not a preferred answer set.

On the other hand $E_2 = \reduct{\lpp{P}}{E_2}$, and $S_2$ is a preferred answer set.

\paragraph{$\snpasindno$}: A generating set $R_1 = \{r_1, r_2, r_3, u_1\}$ corresponds to the answer set $S_1$, and $R_2 = \{r_1, r_2, u_4, u_2\}$ corresponds to the answer set $S_2$.

We have that $\trules{u_4}{R_1} = \{u_1\}$. The rules $r_1, r_2, r_3$ are not included as they are less preferred that $u_4$. Hence $\bodym{u_4} \cap \head{\trules{u_4}{R_1}} = \emptyset$. Therefore $u_4$ cannot be defeated, i.e. $u_4 \in \reduct{\lpp{P}}{R_1}$. Hence $R_1 \neq \minpos{\reduct{\lpp{P}}{R_1}}$, and the answer set $S_1$ is not a preferred answer set.

On the other hand $R_2 = \minpos{\reduct{\lpp{P}}{R_2}}$, and the answer set $S_2$ is a preferred answer set.
\end{example}

\section{Conclusions}

When dealing with preferences it is always important to remember what the abstract term ``preferences'' represents. In this paper we understand preferences as a mechanism for encoding exceptions. In case of conflicting rules, the preferred rules define exceptions to less preferred ones, and not the other way around. For this interpretation of preferences, it is important that a semantics for preferred answer sets satisfies Brewka and Eiter's Principle \ref{principle:i}. All the existing approaches for logic programming with preferences on rules that satisfy the principle introduce an imperative feature into the language. Preferences are understood as the order in which the rules of a program are applied.

The goal of this paper was to develop a purely declarative approach to preference handling satisfying Principle \ref{principle:i}. We have developed two approaches $\snpasindgen$ and $\snpasindno$. The first one is able to ignore preferences between non-conflicting rules. For example, it is equivalent with the answer set semantics on stratified programs. It is designed for situations, where developer does not have full control over preferences. An example is a situation where a user is able to write his/her own rules in order to override developer's rules. If the user's rules are not known until run-time of the system, we have to prefer all the user's rules over the developer's rules. To the best of our knowledge, no existing approach for logic programming with preferences satisfying Principle \ref{principle:i} is usable in this situation. On the other hand, in situations where we can drop the requirement for ignoring preferences between non-conflicting rules, e.g. if a developer has full control over the program, we can use $\snpasindno$ which is in the $NP$ complexity class. Naturally, since the requirement for ignoring preferences between non-conflicting rules was dropped, there are stratified programs with answer sets and no preferred answer sets according to $\snpasindno$.

The two presented approaches are not independent. They form a hierarchy, a branch in the hierarchy of the approaches $\snpasdst$, $\snpaswzl$, $\snpasbe$ and $\snpasd$.

One of our future goals is to better understand the complexity of the decision problem $\pasindgen{\lpp{P}} \neq \emptyset$. So far, we have $\Sigma_3^P$ membership result. It is not immediately clear whether the problem is also $\Sigma_3^P$ hard.

We also plan to investigate relation between $\snpasindgen$ and argumentation, and to implement a prototype solver for the semantics using a meta-interpretation technique of \cite{Eiter:2003uj}.

\section{Acknowledgments}

We would like to thank the anonymous reviewers for detailed and useful comments.
This work was supported by the grant UK/276/2013 of Comenius University in Bratislava and 1/1333/12 of VEGA.

\bibliographystyle{aaai}
\bibliography{library}
\end{document}